\documentclass[12pt,psfig]{article}
\usepackage{multirow}
\usepackage{longtable}
\usepackage{bibentry}
\usepackage{amsfonts}
\usepackage{amssymb}
\usepackage{amsthm}
\usepackage[namelimits]{amsmath}
\usepackage{graphicx,epsfig}
\usepackage{amscd}
\usepackage{bbm}
\usepackage{epsfig,color}
\usepackage{float} 

\date{}

\newcommand{\be}{\beta}

\newcommand{\ba}{\begin{array}}
\newcommand{\ea}{\end{array}}
\newcommand{\beqq}{\begin{equation*}}
\newcommand{\eeqq}{\end{equation*}}
\newcommand{\beq}{\begin{equation}}
\newcommand{\eeq}{\end{equation}}
\newcommand{\bdm}{\begin{displaymath}}
\newcommand{\edm}{\end{displaymath}}

\theoremstyle{definition}
\newtheorem{theorem}{Theorem}[section]
\newtheorem{definition}{Definition}[section]
\newtheorem{lemma}{Lemma}[section]

\newtheorem{remark}{Remark}[section]
\newtheorem{example}{Example}[section]

\numberwithin{equation}{section}

\begin{document}
\pagestyle{plain}

\begin{center}
{\large \bf   Sinai-Ruelle-Bowen measures for piecewise hyperbolic maps with two directions of instability in three-dimensional spaces}
\\ [0.3in]
XU ZHANG \footnote{ Email address:\ \ xuzhang08@gmail.com (X. Zhang).}

 \vspace{0.15in}
{\it  Department of Mathematics, Michigan State University,
East Lansing, MI 48824, USA}
\end{center}

\vspace{0.18in}

\baselineskip=20pt

{\bf\large{ Abstract.}}\ A class of piecewise $C^2$ Lozi-like maps in three-dimensional Euclidean spaces is introduced, and the existence of Sinai-Ruelle-Bowen measures is studied, where the dimension of the instability is equal to two. Further, an example with computer simulations is provided to illustrate the theoretical results.

\section{Introduction}

The study of the existence of the invariant measure of a map is an interesting problem in dynamical systems. Sinai investigated the $C^2$ Anosov diffeomorphism on a compact connected Riemannian manifold, and showed that the measure has absolutely continuous conditional measure on unstable manifolds \cite{Sinai1972}.  Bowen and Ruelle obtained similar results for Axiom A attractors \cite{Bowen1975}. Based on Sinai, Ruelle, and Bowen's work, the invariant Borel measure, which has absolutely continuous conditional measure on unstable manifolds with respect to the Lebesgue measure, is called the Sinai-Ruelle-Bowen measure (SRB measure). For more information on SRB measure, please refer to Young's work \cite{Young2002}.

Later, Pesin developed the non-uniformly hyperbolic theory \cite{BarreiraPesin, Pesin1978}. For singular systems, Katok and Strelcyn investigated the existence of invariant manifolds and obtained some similar results \cite{KatokStre}.  And, there are lots of work on the billiard systems and so on \cite{ChernovMarkarian2006, Ledr84}.
In \cite{JakobsonNewhouse}, Jakobson and Newhouse obtained some sufficient conditions for the existence of SRB measure for piecewise $C^2$ diffeomorphisms with unbounded derivatives. In \cite{Sanchez}, S\'{a}nchez-Salas provided some sufficient conditions for the existence of SRB measure for transformations with infinitely many hyperbolic branches.

In the research of two-dimensional maps, there are two important types of maps, one is the H\'{e}non map, the other is the Lozi map. A series of work on the H\'{e}non map obtained by Benedicks, Carlson, Young, Viana, and Wang, described the relationship between the parameters and the dynamics deeply \cite{BenedicksCarleson, BenedicksViana, BenYou, WangYoung2001}. There are lots of results about the Lozi map \cite{ColletLevy, Young1985}.

The study of the SRB measure also inspires the study of the statistical properties of dynamical systems. The well-known Lasota-Yorke inequality and some generalization contributed to the development of the chaos theory greatly \cite{BoyarskyGora, LasotaYorke}. A powerful tool in the research of the statistical properties is the transfer operator approach \cite{Sarig, Young99, Young}.

In our present work, we apply the bounded variation function method \cite{Liverani2013} and Young's idea in two-dimensional maps \cite{Young1985} to study the existence of SRB measure for a class of three-dimensional maps, which can be thought of as the generalization of the Lozi map in three-dimensional spaces, where the dimension of the  instability is equal to two. Our results and discussions can be easily generalized to study the maps with several directions of instability in high-dimensional spaces.

The rest is organized as follows. In Section 2, the main result is introduced. In Section 3, the proof of the main results is provided. In Section 4, an example with computer simulations is given to illustrate the theoretical results.

\bigskip

\section{Main Results}

In this section, the main results are introduced and some basic concepts and lemmas are given.

Given any $0=a_0<a_1<\cdots<a_p<a_{p+1}=1$ and $0=b_0<b_1<\cdots<b_q<b_{q+1}=1$, denote
$$S_1:=\{a_1,...,a_p\}\times[0,1]\times[0,1],\ S_2:=[0,1]\times\{b_1,...,b_q\}\times[0,1],$$ and
$$S:=S_1\cup S_2,\ R:=[0,1]^3.$$

There is a natural partition of the set $[0,1]^2\setminus \big(\big(\{a_0,a_1,...,a_p,a_{p+1}\}\times[0,1]
\big)\cup\big([0,1]\times\{b_0,b_1,...,b_q,b_{q+1}\}\big)\big)$.
Without loss of generality, suppose this partition is
\beqq
\bigcup^{(p+1)(q+1)}_{k=1}\Omega_k=[0,1]^2\setminus \big(\big(\{a_0,a_1,...,a_p,a_{p+1}\}\times[0,1]
\big)\cup\big([0,1]\times\{b_0,b_1,...,b_q,b_{q+1}\}\big)\big),
\eeqq
where $\Omega_{k_1}\cap\Omega_{k_2}=\emptyset$, if $k_1\neq k_2$, and $\cup\overline{\Omega}_k=[0,1]^2$.
Assume that
\beqq
\Omega_k\cap\{(x,u)\}_{u\in[0,1]}
=\{x\}\times(a_{k}(x),b_{k}(x)),\ \forall x\in[0,1],\ 1\leq k\leq(p+1)(q+1);
\eeqq
and
\beqq \Omega_k\cap\{(u,y)\}_{u\in[0,1]}
=(c_{k}(y),d_{k}(y))\times\{y\},\ \forall y\in[0,1],,\ 1\leq k\leq(p+1)(q+1).
\eeqq
So, for any $ 1\leq k\leq(p+1)(q+1)$, one has
\beqq
\overline{\Omega}_k\cap\{(x,u)\}_{u\in[0,1]}
=\{x\}\times[a_k(x),b_k(x)]\ \mbox{and}\ \overline{\Omega}_k\cap\{(u,y)\}_{u\in[0,1]}
=[c_k(y),d_k(y)]\times\{y\}.
\eeqq

Consider a map $f$ on $R$ such that $f(R)\subset R$ and it satisfies the following assumptions:
\begin{align*}
\mbox{(A0).}\ & f|(R-S)\ \mbox{is a}\ C^2\ \mbox{map},\ f|(R-S)\ \mbox{and}\ f^{-1}|f(R-S)\\
&\mbox{have bounded first and second derivatives, respectively.}
\end{align*}
\begin{equation*}
\mbox{(A1).}\ \inf\bigg\{\bigg|\frac{\partial f_1}{\partial x}\bigg|-\bigg|\frac{\partial f_1}{\partial y}\bigg|-\bigg|\frac{\partial f_1}{\partial z}\bigg|,\ \bigg|\frac{\partial f_2}{\partial y}\bigg|-\bigg|\frac{\partial f_2}{\partial x}\bigg|-\bigg|\frac{\partial f_2}{\partial z}\bigg|\bigg\}=\lambda>1.
\end{equation*}
\begin{equation*}
\mbox{(A2).}\ \sup\bigg\{\bigg|\frac{\partial f_1}{\partial y}\bigg|,\ \bigg|\frac{\partial f_1}{\partial  z}\bigg|,\ \bigg|\frac{\partial f_2}{\partial x}\bigg|,\ \bigg|\frac{\partial f_2}{\partial z}\bigg|,\ \bigg|\frac{\partial f_3}{\partial x}\bigg|,\ \bigg|\frac{\partial f_3}{\partial y }\bigg|,\ \bigg|\frac{\partial f_3}{\partial z}\bigg|\bigg\}\leq\frac{\lambda}{8}.
\end{equation*}
\begin{align*}
\mbox{(A3).}\ &\mbox{There is}\ N\in\mathbb{N}\ \mbox{such that}\ \lambda^N>2\ \mbox{and}\ f_1(f^{k-1}(S_1))\cap \{a_1,...,a_p\}=\emptyset,\\
& f_2(f^{k-1}(S_2))\cap \{b_1,...,b_q\}=\emptyset,\ 1\leq k\leq N.
\end{align*}
The assumption (A1) means that the action of $Df$ when projected onto the $x$-axis and $y$-axis is uniformly expanding, respectively.
\begin{remark}
An example satisfying all the above assumptions is given in the last section.
\end{remark}

\begin{definition} \cite{Young1985}
A Borel probability measures $\mu$ on $R=[0,1]^3$ is said to have absolutely continuous conditional measures on unstable manifolds if there exist measurable partitions $\mathcal{P}_1\subset\mathcal{P}_2\subset\cdots$
of $R$ and measurable sets $V_1\subset V_2\subset V_3\subset\cdots$ such that
\begin{itemize}
\item [(i)] $\mu(V_n)\to1$ as $n\to\infty$;
\item [(ii)] each element of $\mathcal{P}_n|V_n$ is an open subset of some unstable manifold;
\item [(iii)] if $\{\mu_{c}:\ c\in\mathcal{P}_n|V_n\}$ denotes the system of some unstable manifold and $\mathcal{P}_n|V_n$ and $m_{c}$ denotes Riemannian measure on $c$, then for almost every $c\in\mathcal{P}_n|V_n$, one has $\mu_c\ll m_c$.
\end{itemize}
\end{definition}

Now, we introduce the bounded variation functions \cite{EvansGariepy, Liverani2013}. For any $\Omega\subset\mathbb{R}^2$, the support of any function $h=(h_1,h_2)\in L^1(\Omega,m)$ is contained in $\Omega$, where $m$ is the Lebesgue measure. Set
\begin{equation*}
\|h\|_{BV}:=\sup_{\psi\in \mathcal{C}^1_0(\mathbb{R}^2,\mathbb{R}^2),\ \|\psi\|_{\infty}\leq1}\int_{\mathbb{R}^2}
h\cdot \mbox{div}\psi\; dm,
\end{equation*}
where $\psi=(\psi_1,\psi_2)$, $\mbox{div}\psi=\partial_x\psi_1+\partial_y\psi_2$, and $\mathcal{C}^r_0$ represents the vector space of $r$-times differentiable functions with compact support. The bounded variation functions are a subset of $L^1$ with $\|\cdot\|_{BV}$ finite.

\begin{lemma}\cite{Liverani2013}\label{bvinequ}
\begin{itemize}
\item [(i)] There exists a constant $C_0>0$ such that
    \begin{equation*}
\|h\|_{L^1}\leq\|h\|_{L^{2}}\leq C_0\|h\|_{BV},\ \forall h\in BV(\Omega);
\end{equation*}
\item [(ii)]
for almost $x$ and $y$, $$h(x,\cdot),\ h(\cdot,y)\in BV(\mathbb{R})\subset L^{\infty}(\mathbb{R},m),\ \forall h\in BV(\Omega);$$
\item [(iii)] for each $h\in BV(\mathbb{R}^2)$ and all $\psi\in L^{\infty}(\mathbb{R}^2,m)$ of compact support and such that almost surely, $\psi(x,\cdot)$, $\psi(\cdot,y)\in\mathcal{C}^0(\mathbb{R},\mathbb{R})$, $\partial_{x}\psi(\cdot,y)$,  $\partial_y\psi(x,\cdot)\in L^1(\mathbb{R},m)$, one has
\begin{equation*}
\bigg|\int_{\mathbb{R}^2}h\cdot\mbox{div}\psi\; dm\bigg|\leq \|h\|_{BV}\|\psi\|_{L^{\infty}}.
\end{equation*}
\end{itemize}
\end{lemma}

The main result is stated as follows:

\begin{theorem}\label{mainresult}
For the map $f$ satisfying (A0)--(A3), there exists an invariant measure, which is an SRB measure.
\end{theorem}

\bigskip

\section{The existence of SRB measure}

In this section, it is to show Theorem \ref{mainresult}.

It follows from the definition of the map $f$ that the unstable manifold are piecewise smooth surface zigzag across $R$, which are turning around at unknown places. To avoid the singular set, the strategy is to construct an invariant measure $\mu$ with good dynamical behavior on a neighborhood of the singular set $S$.

First, given any $C^2$ surface $\alpha:[0,1]\times[0,1]\to[0,1]$, suppose that (A1) and (A2) hold, it is to show that if the angle between the normal vector of the surface and the $z$-axis (including both the positive and negative axes) is less than $45$ degrees, then the angle between the normal vector of $f(\mbox{graph}(\alpha))$ and the $z$-axis is also less than $45$ degrees, except points in the image of the singular set.

The graph of $\alpha$ is $(x,y,\alpha(x,y))$. The normal vector is the cross product of the vectors $<1,0,\alpha_x>$ and $<0,1,\alpha_y>$, that is, $<-\alpha_x,-\alpha_y,1>$. The cosine of the angle between the normal vector and the $z$-axis is $\frac{1}{\sqrt{1+\alpha^2_x+\alpha^2_y}}$ or $\frac{-1}{\sqrt{1+\alpha^2_x+\alpha^2_y}}$. The assumption that the angle between the normal vector and the $z$-axis is less than $45$ degrees is equivalent to
\beq\label{angle45}
\bigg|\frac{1}{\sqrt{1+\alpha^2_x+\alpha^2_y}}\bigg|\geq\frac{\sqrt{2}}{2}.
\eeq
Since $f(\mbox{graph}(\alpha))=(f_1(x,y,\alpha),f_2(x,y,\alpha),f_3(x,y,\alpha))$, one has that the tangent vectors are
\beqq
\bigg(\frac{\partial f_1}{\partial x}+\frac{\partial f_1}{\partial z}\frac{\partial \alpha}{\partial x}, \frac{\partial f_2}{\partial x}+\frac{\partial f_2}{\partial z}\frac{\partial \alpha}{\partial x},
\frac{\partial f_3}{\partial x}+\frac{\partial f_3}{\partial z}\frac{\partial \alpha}{\partial x} \bigg)
\eeqq
and
\beqq
\bigg(\frac{\partial f_1}{\partial y}+\frac{\partial f_1}{\partial y}\frac{\partial \alpha}{\partial y}, \frac{\partial f_2}{\partial y}+\frac{\partial f_2}{\partial z}\frac{\partial \alpha}{\partial y},
\frac{\partial f_3}{\partial y}+\frac{\partial f_3}{\partial z}\frac{\partial \alpha}{\partial y} \bigg).
\eeqq
So, the cross product is
\begin{align*}
&\bigg[\bigg(\frac{\partial f_2}{\partial x}\frac{\partial f_3}{\partial y}-\frac{\partial f_3}{\partial  x}\frac{\partial f_2}{\partial y }\bigg)+\bigg(\frac{\partial f_2}{\partial x}\frac{\partial f_3}{\partial z}-\frac{\partial f_3}{\partial x}\frac{\partial f_2}{\partial z }\bigg)\frac{\partial \alpha}{\partial y}+\bigg(\frac{\partial f_3}{\partial y}\frac{\partial f_2}{\partial z}-\frac{\partial f_2}{\partial y}\frac{\partial f_3}{\partial z}\bigg)\frac{\partial \alpha}{\partial x}\bigg]\vec{i} \\
-&\bigg[\bigg(\frac{\partial f_1}{\partial x}\frac{\partial f_3}{\partial y}-\frac{\partial f_3}{\partial  x}\frac{\partial f_1}{\partial y }\bigg)+\bigg(\frac{\partial f_1}{\partial x}\frac{\partial f_3}{\partial z}-\frac{\partial f_3}{\partial x}\frac{\partial f_1}{\partial z }\bigg)\frac{\partial \alpha}{\partial y}+\bigg(\frac{\partial f_3}{\partial y}\frac{\partial f_1}{\partial z}-\frac{\partial f_1}{\partial y}\frac{\partial f_3}{\partial z}\bigg)\frac{\partial \alpha}{\partial x}\bigg]\vec{j} \\
+&\bigg[\bigg(\frac{\partial f_1}{\partial x}\frac{\partial f_2}{\partial y}-\frac{\partial f_2}{\partial  x}\frac{\partial f_1}{\partial y }\bigg)+\bigg(\frac{\partial f_1}{\partial x}\frac{\partial f_2}{\partial z}-\frac{\partial f_2}{\partial x}\frac{\partial f_1}{\partial z }\bigg)\frac{\partial \alpha}{\partial y}+\bigg(\frac{\partial f_2}{\partial y}\frac{\partial f_1}{\partial z}-\frac{\partial f_1}{\partial y}\frac{\partial f_2}{\partial z}\bigg)\frac{\partial \alpha}{\partial x}\bigg]\vec{k}\\
:=&A\vec{i}+B\vec{j}+C\vec{k}.
\end{align*}

The absolute value of the cosine of the angle between the normal vector and the $z$-axis is $\frac{|C|}{\sqrt{A^2+B^2+C^2}}$. By (A2), (A3), and \eqref{angle45}, one has that $|C|\geq|A|+|B|$, which implies that $\frac{|C|}{\sqrt{A^2+B^2+C^2}}\geq\frac{\sqrt{2}}{2}.$ Hence, the angle between the normal vector and the $z$-axis is less than $45$ degrees.

Let $p_x:R\to[0,1]$ and $p_y:R\to[0,1]$ be the projection onto the $x$-axis and $y$-axis, respectively. The Lebesgue measure on $[0,1]\times[0,1]$ is denoted by $m$. If $\mu$ is a measure on $R$, then $f_*\mu$ is defined by $f_*\mu(E)=\mu(f^{-1}(E))$.
Let $J=[a,b]\times[c,d]\subset[0,1]\times[0,1]$ be a closed rectangle and $\alpha:J\to[0,1]$ be a $C^2$ function with the normal vector very close to the $z$-axis. Then the image of  $\mbox{graph}(\alpha)$ under the map $f$ is a union of finitely many smooth surfaces, which are denoted by $\{L_i(\alpha)\}$. Similarly, set the smooth surfaces of $f^k(\mbox{graph}(\alpha))$ as $\{L_{i_1i_2\cdots i_k}\}$ such that $f(L_{i_1i_2\cdots i_k})=\cup_jL_{i_1i_2\cdots i_kj}$. Let $\mu_0$ be the measure on $\mbox{graph}(\alpha)$ such that $(p_x\times p_y)_*\mu_0$ is the normalized Lebesgue measure on $J$, set $\mu_k:=(f^k)_*\mu_0$.  By (A1) and (A2), one has that $(p_x\times p_y)_*\mu_k|L_{i_1i_2\cdots i_k}$ is absolutely continuous with respect to Lebesgue measure $m$. The density of $(p_x\times p_y)_*\mu_k|L_{i_1i_2\cdots i_k}$ and $(p_x\times p_y)_*\mu_k$ are denoted by $\rho_{i_1i_2\cdots i_k}$ and $\hat{\rho}_k$, respectively. So, one has that $\sum_{i_1i_2\cdots i_k}\rho_{i_1i_2\cdots i_k}=\hat{\rho}_k$.

If $N>1$, where $N$ is specified in (A3), we could define the sets $L_{i_1i_2\cdots i_k}$ and $L_{i_1i_2\cdots i_kj}$ as above for the map $f^N$. Without loss of generality, assume that $N=1$ in the following discussions, that is, $\lambda>2$.

In the following discussions, fix $L_{i_1i_2\cdots i_k}$ and $L_{i_1i_2\cdots i_kj}$. Set $(p_x\times p_y)L_{i_1\cdots i_k}=\cup^s_{j=1}B_j$, where each $B_j$ is contained in some $\overline{\Omega}_{j'}$, such that $\Omega_{j'_1}\cap\Omega_{j'_2}
=\emptyset$, where $B_{j'_1}\subset\overline{\Omega}_{j'_1}$, $B_{j'_2}\subset\overline{\Omega}_{j'_2}$, and $j'_1\neq j'_2$. So, for the given $j$, there is a unique $B_{j+}$ such that $f(B_{j+})=L_{i_1\cdots i_kj}$. Set
$D_{j-}=(p_x\times p_y)L_{i_1\cdots i_kj}$. Define the following map $T:B_{j+}\to D_{j-}$:
\beqq
T(x,y):=(p_x\times p_y)\circ f \circ (p_x\times p_y)^{-1}(x,y).
\eeqq
By (A0) and (A1), one has that $T$ is $C^2$ between $B_{j+}$ and $D_{j-}$,
$$\bigg|\frac{\partial T_1}{\partial x}\bigg|-\bigg|\frac{\partial T_{1}}{\partial y}\bigg|\geq \lambda\ \mbox{and}\ \bigg|\frac{\partial T_{2}}{\partial y}\bigg|-\bigg|\frac{\partial T_{2}}{\partial x}\bigg|\geq \lambda,$$
where $T=(T_1,T_2)$. Set $T^{-1}:=(T^{-1}_1,T^{-1}_2)$.

Next, it is to study the density of the invariant measure on the unstable manifolds. And, it is to show the following lemma.

\begin{lemma} \label{densitybv}
For the given surface $\alpha$ as above, there exist an invariant Borel probability measure $\mu$ on $\alpha$ and a function $\rho:[0,1]\times[0,1]\to\mathbb{R}$ of bounded variation such that $d((p_x\times p_y)_*\mu)=\rho dm$.
\end{lemma}

\begin{proof}
For the given surface $\alpha$, it is to show that there is a positive constant $M$ such that $\|\hat{\rho}_k\|_{BV([0,1]^2)}\leq M$.

Suppose $\phi=(\phi_1,\phi_2)$ with $\phi\in \mathcal{C}^1_0(\mathbb{R}^2,\mathbb{R}^2)$ and
$\|\phi\|_{\infty}\leq1$.
It follows from direct calculation that
\begin{align*}
&\sum_j\|\rho_{i_1\cdots i_kj}\|_{BV}=\sum_j\int^1_0\int^1_0\rho_{i_1\cdots i_kj}(x,y)\bigg(\frac{\partial\phi_1(x,y)}{\partial x}+\frac{\partial\phi_2(x,y)}{\partial y}\bigg)dxdy\\
=&\int^1_0\int^1_0\frac{\rho_{i_1\cdots i_k}(T^{-1}_1(x,y),T^{-1}_2(x,y))}{\mbox{det}(DT(T^{-1}_1(x,y),T^{-1}_2(x,y)))}
\bigg(\frac{\partial\phi_1(x,y)}{\partial x}+\frac{\partial\phi_2(x,y)}{\partial y}\bigg)dxdy\\
=&\int^1_0\int^1_0 \rho_{i_1\cdots i_k}(x,y) \bigg(\frac{\partial\phi_1}{\partial x}\circ T+\frac{\partial\phi_2}{\partial y}\circ T\bigg)dxdy.
\end{align*}

By direct computation, one has
\beqq
\left(\begin{array}{c}
  \frac{\partial (\phi_1\circ T)}{\partial x} \\
  \frac{\partial (\phi_1\circ T)}{\partial y}
\end{array}
 \right)
=\left(
   \begin{array}{cc}
     \frac{\partial T_1}{\partial x} & \frac{\partial T_2}{\partial x} \\
     \frac{\partial T_1}{\partial y} & \frac{\partial T_2}{\partial y} \\
   \end{array}
 \right)
\left(
\begin{array}{c}
  \frac{\partial\phi_1}{\partial x}\circ T \\
  \frac{\partial\phi_1}{\partial y}\circ T
\end{array}
 \right),
\eeqq
\beqq
\left(\begin{array}{c}
  \frac{\partial (\phi_2\circ T)}{\partial x} \\
  \frac{\partial (\phi_2\circ T)}{\partial y}
\end{array}
 \right)
=\left(
   \begin{array}{cc}
     \frac{\partial T_1}{\partial x} & \frac{\partial T_2}{\partial x} \\
     \frac{\partial T_1}{\partial y} & \frac{\partial T_2}{\partial y} \\
   \end{array}
 \right)
\left(
\begin{array}{c}
  \frac{\partial\phi_2}{\partial x}\circ T \\
  \frac{\partial\phi_2}{\partial y}\circ T
\end{array}
 \right),
\eeqq
which implies that
\begin{align*}
&\frac{\partial\phi_1}{\partial x}\circ T+\frac{\partial\phi_2}{\partial y}\circ T\\
=&\bigg([(DT)^{-1}]_{11}\frac{\partial(\phi_1\circ T)}{\partial x}+[(DT)^{-1}]_{21}\frac{\partial(\phi_2\circ T)}{\partial x}\bigg)\\
&+\bigg([(DT)^{-1}]_{12}\frac{\partial(\phi_1\circ T)}{\partial y}+
[(DT)^{-1}]_{22}\frac{\partial(\phi_2\circ T)}{\partial y}\bigg)\\
=&\frac{\partial}{\partial x}\bigg([(DT)^{-1}]_{11}\phi_1\circ T+[(DT)^{-1}]_{21}\phi_2\circ T\bigg)\\
&+\frac{\partial}{\partial y}\bigg([(DT)^{-1}]_{12}\phi_1\circ T+[(DT)^{-1}]_{22}\phi_2\circ T\bigg)\\
&-\bigg(\phi_1\circ T\frac{\partial}{\partial x}[(DT)^{-1}]_{11}+\phi_2\circ T\frac{\partial}{\partial x}[(DT)^{-1}]_{21}\\
&+\phi_1\circ T\frac{\partial}{\partial y}[(DT)^{-1}]_{12}+\phi_2\circ T\frac{\partial}{\partial y}[(DT)^{-1}]_{22}\bigg).
\end{align*}
where
\beq
DT=\left(
   \begin{array}{cc}
     \frac{\partial T_1}{\partial x} & \frac{\partial T_2}{\partial x} \\
     \frac{\partial T_1}{\partial y} & \frac{\partial T_2}{\partial y} \\
   \end{array}
 \right)
\eeq
and
\beq\label{matrixinverse}
DT^{-1}=\frac{1}{\partial_{x}T_1\partial_{y}T_2
-\partial_{x}T_2\partial_{y}T_1}
\left(
   \begin{array}{cc}
     \frac{\partial T_2}{\partial y} & -\frac{\partial T_2}{\partial x} \\
     -\frac{\partial T_1}{\partial y} & \frac{\partial T_1}{\partial x} \\
   \end{array}
 \right)=
  \left(
     \begin{array}{cc}
       [(DT)^{-1}]_{11} & [(DT)^{-1}]_{12} \\
        \mbox{[}(DT)^{-1}]_{21} & [(DT)^{-1}]_{22}\\
     \end{array}
   \right).
\eeq

Given any $k$, $1\leq k\leq (p+1)(q+1)$,  for any $(x,y)\in[0,1]^2\setminus((\{a_1,...,a_p\}\times[0,1])\cup([0,1]\times\{b_1,...,b_q\}))$, set
\beqq
\Phi_{1,k}(x,y):=[(DT)^{-1}]_{11}\phi_1\circ T(x,y)\mathbbm{1}_{p_x(\overline{\Omega}_k)}(x)+[(DT)^{-1}]_{21}\phi_2\circ T(x,y)\mathbbm{1}_{p_y(\overline{\Omega}_k)}(y),
\eeqq
\beqq
\Phi_{2,k}(x,y):=[(DT)^{-1}]_{12}\phi_1\circ T(x,y)\mathbbm{1}_{p_x(\overline{\Omega}_k)}(x)+
[(DT)^{-1}]_{22}\phi_2\circ T(x,y)\mathbbm{1}_{p_y(\overline{\Omega}_k)}(y),
\eeqq
\beqq
\Phi^{-}_{1,k}(y):=\Phi_{1,k}(c_k(y),y),\ \Phi^{+}_{1,k}(y):=\Phi_{1,k}(d_k(y),y),
\eeqq
\beqq
 \Phi^{-}_{2,k}(x):=\Phi_{2,k}(x,a_k(x)),\ \Phi^{+}_{2,k}(x):=\Phi_{2,k}(x,b_k(x)).
\eeqq

It follows from \eqref{matrixinverse}, (A2), and (A3) that there exist $\delta>0$ and $\tau>1$ such that
\begin{align} \label{bound}
&\sup_{(x,y)\in[0,1]^2}\sup
\bigg\{\sum_{k:[x-\delta,x+\delta]\cap[c_k(y),d_k(y)]\neq\emptyset}\|\mathbbm{1}_{[c_k(y),d_k(y)]}(\cdot)[(DT)^{-1}](\cdot, y)\|_{\infty},\nonumber\\
&\sum_{k:[y-\delta,y+\delta]\cap[a_k(x),b_k(x)]\neq\emptyset}\|\mathbbm{1}_{[a_k(x),b_k(x)]}(\cdot)[(DT)^{-1}](x,\cdot)\|_{\infty}\bigg\}\leq \tau^{-1}.
\end{align}
Denote
\beqq
\eta_{1,k,y}(v):=\left\{
  \begin{array}{ll}
    0, & \hbox{if}\ v\in(-\infty,c_k(y)-\delta) \\
    \Phi^{-}_{1,k}(y)(v-c_k(y)+\delta)\delta^{-1}, & \hbox{if}\ v\in[c_k(y)-\delta,c_k(y)) \\
    0, & \hbox{if}\ v\in[c_k(y),d_k(y)] \\
   \Phi^{-}_{1,k}(y)(d_k(y)+\delta-v)\delta^{-1}, & \hbox{if}\ v\in(d_k(y),d_k(y)+\delta) \\
    0, & \hbox{if}\ v\in[d_k(y)+\delta,+\infty),
  \end{array}
\right.
\eeqq
\beqq
\eta_{2,k,x}(v):=\left\{
  \begin{array}{ll}
    0, & \hbox{if}\ v\in(-\infty,a_k(x)-\delta) \\
    \Phi^{-}_{2,k}(x)(v-a_k(x)+\delta)\delta^{-1}, & \hbox{if}\ v\in[a_k(x)-\delta,a_k(x)) \\
    0, & \hbox{if}\ v\in[a_k(x),b_k(x)] \\
   \Phi^{+}_{2,k}(x)(b_k(x)+\delta-v)\delta^{-1}, & \hbox{if}\ v\in(b_k(x),b_k(x)+\delta) \\
    0, & \hbox{if}\ v\in[b_k(x)+\delta,+\infty).
  \end{array}
\right.
\eeqq

Set
\beqq
\bar{\Phi}_{1,k}(x,y):=\Phi_{1,k}(x,y)+\eta_{1,k,y}(x),\ \bar{\Phi}_{2,k}(x,y):=\Phi_{2,k}(x,y)+\eta_{2,k,x}(y),
\eeqq
\beqq
\Theta_1(x,y):=\sum_k\bar{\Phi}_{1,k}(x,y),\ \Theta_2(x,y):=\sum_k\bar{\Phi}_{2,k}(x,y).
\eeqq
By the construction above, $\bar{\Phi}_{1,k}$ and $\bar{\Phi}_{2,k}$ are continuous functions, and for $(x,y)\in[0,1]^2\setminus((\{a_1,...,a_p\}\times[0,1])\cup([0,1]\times\{b_1,...,b_q\}))$,  by \eqref{bound} and $\|\phi\|_{\infty}\leq1$, one has
\begin{align*}
&\max_{j=1,2}|\Theta_j(x,y)|\leq\sup\bigg\{\sum_k\bigg
(\sup_{x\in[c_k(y),d_k(y)]}|\Phi_{1,k}(x,y)|\bigg)\mathbbm{1}_{[c_k(y),d_k(y)]}(x),\\
&\sum_k\bigg
(\sup_{y\in[a_k(x),b_k(x)]}|\Phi_{2,k}(x,y)|\bigg)\mathbbm{1}_{[a_k(x),b_k(x)]}(y)\bigg\}\leq\tau^{-1}.
\end{align*}

Further, set
\beqq
\Theta_1(x,y):=
\sum_k\int^{x}_0\partial_v\bar{\Phi}_{1,k}(v,y)dv=\int^{x}_0\sum_k\partial_v\bar{\Phi}_{1,k}(v,y)dv,
\eeqq
\beqq
\Theta_2(x,y):=
\sum_k\int^{y}_0\partial_v\bar{\Phi}_{2,k}(x,v)dv=\int^{y}_0\sum_k\partial_v\bar{\Phi}_{2,k}(x,v)dv.
\eeqq
Hence, one has that
\begin{align*}
&\mbox{div}\Theta(x,y)=
\partial_x\Theta_1+\partial_y\Theta_2\\
=&\sum_k\bigg(\partial_x\{[(DT)^{-1}]_{11}\phi_1\circ T+[(DT)^{-1}]_{21}\phi_2\circ T\}\mathbbm{1}_{[c_k(y),d_k(y)]}\\
&+\partial_y\{[(DT)^{-1}]_{12}\phi_1\circ T+[(DT)^{-1}]_{22}\phi_2\circ T\}\mathbbm{1}_{[a_k(x),b_k(x)]}\bigg)\\
&+\delta^{-1}\sum_k\bigg(\Phi^{-}_{1,k}(y)\mathbbm{1}_{[a_k(x),b_k(x)]}-\Phi^{+}_{1,k}(y)\mathbbm{1}_{[a_k(x),b_k(x)]}\\
&+\Phi^{-}_{2,k}(x)\mathbbm{1}_{[c_k(y),d_k(y)]}-\Phi^{+}_{2,k}(x)\mathbbm{1}_{[c_k(y),d_k(y)]}\bigg).
\end{align*}
Thus, by Lemma \ref{bvinequ}, one has that
\begin{align*}
&\sum_j\int^1_0\int^1_0\rho_{i_1\cdots i_kj}\cdot\mbox{div}\phi dxdy\leq \int^1_0\int^1_0\rho_{i_1\cdots i_k}\cdot\mbox{div}\Theta dxdy+\frac{2\|\rho_{i_1\cdots i_k}\|_1}{\tau\delta}\\
&+\bigg|\int^1_0\int^1_0\rho_{i_1\cdots i_k}\sum_k \bigg(\phi_1\circ T\frac{\partial}{\partial x}[(DT)^{-1}]_{11}+\phi_2\circ T\frac{\partial}{\partial x}[(DT)^{-1}]_{21}\\
&+\phi_1\circ T\frac{\partial}{\partial y}[(DT)^{-1}]_{12}+\phi_2\circ T\frac{\partial}{\partial y}[(DT)^{-1}]_{22}\bigg)dxdy\bigg|\\
\leq& \frac{\|\rho_{i_1\cdots i_k}\|_{BV([0,1]^2)}}{\tau}+\frac{2\|\rho_{i_1\cdots i_k}\|_1}{\tau\delta}+C_0\|\rho_{i_1\cdots i_k}\|_1\\
\leq&\frac{\|\rho_{i_1\cdots i_k}\|_{BV([0,1]^2)}}{\tau}+C_1\|\rho_{i_1\cdots i_k}\|_1.
\end{align*}
Set $\beta_k:=\sum_{i_1\cdots i_k}\|\rho_{i_1\cdots i_k}\|_{BV([0,1]^2)}$. One has that
\beqq
\beta_{k+1}=\sum_{i_1\cdots i_k}\bigg(\sum_j\|\rho_{i_1\cdots i_kj}\|_{BV}\bigg)\leq C_1+\frac{1}{\tau}\be_k.
\eeqq
Hence, for any $k$,
\beqq
\|\hat{\rho}_k\|_{BV([0,1]^2)}\leq \be_k\leq C_1\sum^{\infty}_{i=0}\bigg(\frac{1}{\tau}\bigg)^i:=M<\infty.
\eeqq

Thus, one has that
\beqq
\bigg\|\frac{1}{n}\sum^n_{k=1}\hat{\rho}_k\bigg\|_{BV}\leq M.
\eeqq
Hence, it follows from Lemma \ref{bvinequ} that the sequence $\{n^{-1}\sum^n_{k=1}\hat{\rho}_k\}_{n\in\mathbb{N}}$ is precompact in $L^1([0,1]^2,m)$. There exists a convergent subsequence, denoted by $\rho$, the corresponding measure is convergent in the weak star topology, which is a Borel probability measure $\mu$.

This completes the proof.
\end{proof}

By Lemma \ref{densitybv}, one has
$$
\sum^{\infty}_{k=0}\mu(f^kD(S,\delta \lambda^{-k}))=\sum^{\infty}_{k=0}\mu(D(S,\delta \lambda^{-k}))\leq 4(p+1)(q+1)\delta M\sum^{\infty}_{k=0}\lambda^{-k}<\infty,
$$
where $D(S,\delta)$ is the $\delta$-neighborhood of the  singular set $S$, $M$ is specified in the proof of Lemma \ref{densitybv}.
It follows from the Borel-Cantelli Lemma that $w$ is in $f^kD(S,\delta\lambda^{-k})$ for at most finitely many $k$, $\mu$-a.e., that is, for $\mu$ almost everywhere $w$, there is $\delta(w)>0$ such that $f^{-k}w\not\in D(S,\delta(w)\lambda^{-k})$ for all $k>0$, implying the existence of local unstable manifold $W^u_{\delta(w)}(w)$ by \cite{KatokStre}.

Now, it is to show Theorem \ref{mainresult}.

Pick some $w$ such that $\mbox{graph}(\alpha)=W^u_{\delta(w)}(w)$ exists. Fix this $W^u_{\delta(w)}(w)$ as a smooth surface $\alpha$. For $\delta>0$, set $\Lambda_{\delta}:=\{w\in R:\ d(f^{-k}w,S)\geq\delta\lambda^{-k},\ \forall k\geq0\}$ and $\Lambda_0:=\lim_{\delta\to0}\mu(\Lambda_{\delta})$ .

Next, it is to define a sequence of measurable partitions $\mathcal{P}_1\subset\mathcal{P}_2
\subset\mathcal{P}_3\subset\cdots$. For any $n\in\mathbb{N}$, let $\{U_{i,j}\}_{1\leq i,j\leq 2^n}$ be a partition of $R$, where $U_{i,j}=\{(x,y,z):\ \frac{i-1}{2^n}\leq x\leq\frac{i}{2^n}, \frac{j-1}{2^n}\leq y\leq\frac{j}{2^n},0\leq z\leq 1\}$. For $w\in U_{i,j}\cap \Lambda_{1/2^{n}}$, set $c(w):=W^u_{1/2^{n}}(w)\cap U_{i,j}$, $V_n:=\cup_{w\in\Lambda_{1/2^{n}}}c(w)$, and $\mathcal{P}_n:=\{c(w):\ w\in V_n\}\cup\{R-V_n\}$.

Fix a partition $\mathcal{P}_n$, it is to define a sequence of measures $\{\tilde{\mu}_k\}_{k\in\mathbb{N}}$ as follows: since $\mu_k=f^k_*\mu_0$ is defined on $f^k(\mbox{graph}(\alpha))$ and $f^k(\mbox{graph}(\alpha))$ is a finite union of smooth surfaces, let $\tilde{\mu}_k$ be $\mu_k$ annihilated on those parts of its support that only partially cross some $\tilde{U}_{i,j}$, that is, the support of $\tilde{\mu}_k$ consists of all of the sets $V_0$ satisfying that there is a smooth component of $f^k(\mbox{graph}(\alpha))$, denoted by $W_0$, such that $V_0\subset W_0$ and $(p_x\times p_y)V_0=\cup_{i_0,j_0}(p_x\times p_y)(U_{i_0,j_0})$.

Next, it is to show that given any $\epsilon>0$, there is $n=n(\epsilon)>0$ such that for the fixed partition $\mathcal{P}_n$ and sufficiently large $k$, $\tilde{\mu}_k(R)>1-\epsilon$. For $w\in(\mbox{support}(\mu_k)
-\mbox{support}(\tilde{\mu}_k))$, $w$ is either in a small piece, which only partially crosses some $\tilde{U}_{i,j}$, or the distance between $w$ and a cusp in $f^k(\mbox{graph}(\alpha))$ is less than $1/2^{n(\epsilon)}$. Hence, one has
\begin{align*}
&1-\tilde{\mu}_k(R)=\mu_k(\mbox{support}(\mu_k)
-\mbox{support}(\tilde{\mu}_k))\\
=&\sum^k_{i=1}\mu_i\{w:\ d(f^{-i}w,S)\leq2^{-n(\epsilon)}\lambda^{-i}\}+\mu_k(\{\mbox{the boundary of }f^k(\mbox{graph}(\alpha))\})\\
=&\sum^k_{i=1}\mu_i\{w:\ d(fw,S)\leq2^{-n(\epsilon)}\lambda^{-i}\}+\mu_k(\{\mbox{the boundary of }f^k(\mbox{graph}(\alpha))\})\\
\leq& 4M(p+1)(q+1)2^{-n(\epsilon)}\sum^k_{i=1}\lambda^{-i}+\mu_k(\{\mbox{the boundary of }f^k(\mbox{graph}(\alpha))\}),
\end{align*}
where $M$ is specified in the proof of Lemma \ref{densitybv}.
The first term is very close to zero as $n(\epsilon)$ goes to positive infinite, the second term goes to zero as $k$ goes to positive infinite. For the given $\epsilon$, take a sufficiently large $n(\epsilon)$, the corresponding partition $\mathcal{P}_{n(\epsilon)}$ is denoted by $\tilde{U}_{i,j}$, $1\leq i,j\leq 2^{n(\epsilon)}$.

Since $\lim_{i\to\infty}\frac{1}{n_i}\sum^{n_i}_{k=1}\mu_k\to\mu$ in the weak topology, there exists a subsequence $\{n_{i_j}\}$ of $\{n_i\}$ such that $\lim_{j\to\infty}\frac{1}{n_{i_j}}\sum^{n_{i_j}}_{k=1}
\tilde{\mu}_k\to\tilde{\mu}$ in the weak topology. It follows from the definitions of $\tilde{\mu}_k$ and $\mu_k$ that one has that $\tilde{\mu}\ll\mu$ and $0\leq d\tilde{\mu}/d\mu\leq1$. So, by taking $n(\epsilon)$ large enough, one has that $\tilde{\mu}$ is equivalent to $\mu$ except on a set with the $\mu$-measure less than $\epsilon$.

It is to show that there is a transverse measure $\tilde{\mu}_T$ for the measure $\tilde{\mu}$ such that
for $\tilde{\mu}_T$-a.e. $c\in\mathcal{P}_{n(\epsilon)}$, one has that $\tilde{\mu}_c\ll m_c$.
Denote $\tilde{g}_{i_1\cdots i_k}$ as the density of $(p_x\times p_y)_*(\tilde{\mu}_k|L_{i_1\cdots i_k})$. For $1\leq i,j\leq 2^{n(\epsilon)}$, $\tilde{U}_{i,j}$, any $i_1,...,i_k$, it is to show that either
\beqq
\tilde{g}_{i_1\cdots i_k}=0\ \mbox{on}\ (p_x\times p_y)\tilde{U}_{i,j},
\eeqq
or
\beqq
\frac{\tilde{g}_{i_1\cdots i_k}(w_1)}{\tilde{g}_{i_1\cdots i_k}(w_2)}\leq M_1,\ \forall w_1,w_2\in(p_x\times p_y)\tilde{U}_{i,j},
\eeqq
where $M_1$ is a positive constant independent on the choice of $i_1,...,i_k$. If $L_{i_1,,,i_k}$ does not cross the full $\tilde{U}_{ij}$, then $\tilde{g}_{i_1\cdots i_k}=0$ on $(p_x\times p_y)\tilde{U}_{i,j}$. Suppose that there are subsets $E_0,E_1,...,E_k\subset [0,1]\times[0,1]$ with $E_k\subset(p_x\times p_y)\tilde{U}_{i,j}$, and diffeomorphic map $h_l:E_{l-1}\to E_l$ such that $\tilde{g}_{i_1\cdots i_k}|(p_x\times p_y)\tilde{U}_{i,j}=\mbox{det}((h_k\circ\cdots h_1)^{-1})$.

Suppose $w_1=(x_1,y_1)$ and $w_2=(x_2,y_2)$. It follows from the assumption that $f|(R-S)$ has bounded first and second derivatives that
\begin{align*}
&|\mbox{det}DT(w_1)-\mbox{det}DT(w_2)|\\
=&\bigg|\mbox{det}\left(\begin{array}{cc}
\partial_xT_1(w_1) & \partial_yT_1(w_1)\\
\partial_xT_2(w_1) & \partial_yT_2(w_1)
\end{array}
\right)-\mbox{det}\left(\begin{array}{cc}
\partial_xT_1(w_2) & \partial_yT_1(w_2)\\
\partial_xT_2(w_2) & \partial_yT_2(w_2)
\end{array}
\right)\bigg|\\
=& |(\partial_xT_1(w_1) \partial_yT_2(w_1)-
\partial_yT_1(w_1)\partial_xT_2(w_1))\\
&-
(\partial_xT_1(w_2) \partial_yT_2(w_2)-
\partial_yT_1(w_2)\partial_xT_2(w_2))|\\
\leq&|\partial_xT_1(w_1)(\partial_yT_2(w_1)-
\partial_yT_2(w_2)|+|\partial_yT_2(w_2)(
\partial_xT_1(w_1)-\partial_xT_1(w_2))|\\
&+|\partial_yT_1(w_1)(\partial_xT_2(w_1)-\partial_xT_2(w_2))|+|\partial_xT_2(w_2)(
\partial_yT_1(w_1)-\partial_yT_1(w_2))|\\
\leq&C_1|w_1-w_2|,
\end{align*}
where $C_1$ is a positive constant. Further, since $\bigg|\frac{\partial T_1}{\partial x}\bigg|-\bigg|\frac{\partial T_{1}}{\partial y}\bigg|\geq \lambda$,
 $\bigg|\frac{\partial T_{2}}{\partial y}\bigg|-\bigg|\frac{\partial T_{2}}{\partial x}\bigg|\geq \lambda$, (A0), and
\eqref{matrixinverse}, it follows from direct calculation that
\beqq
|\log \mbox{det}DT^{-1}(w_1)-\log\mbox{det}DT^{-1}(w_2)|\leq C_2|w_1-w_2|,
\eeqq
where $C_2$ is a positive constant. The last inequality implies that
\beqq
\frac{\tilde{g}_{i_1\cdots i_k}(w_1)}{\tilde{g}_{i_1\cdots i_k}(w_2)}\leq M_1,\ \forall w_1,w_2\in(p_x\times p_y)\tilde{U}_{i,j},
\eeqq
where $M_1$ is independent on $i_1,...,i_k$.

This completes the whole proof.

\medskip

\begin{remark}
It is easy to obtain similar results for maps with several directions of instability in high-dimensional spaces.
\end{remark}

\begin{remark}
Some statistical properties of these maps could be obtained by using the transfer operator methods with some proper function spaces.
\end{remark}
\bigskip

\section{Example}

In this section, we give an example to illustrate the results in Theorem \ref{mainresult} by computer simulations. The software Mathematica draws the pictures of the simulations.

\begin{example}
Consider the map $f=(f_1(x,y,z),f_2(x,y,z),f_3(x,y,z)):\mathbb{R}^3\to\mathbb{R}^3$, where
\beqq
f_1(x,y,z)=\left\{
\begin{array}{ll}
k_1x+k_2(y+z)& \hbox{if}\ x\in[0,1/3]\\
-k_1(x-1/3)+k_1/3+k_2(y+z)& \hbox{if}\ x\in[1/3,2/3]\\
k_1(x-2/3)+k_2(y+z)& \hbox{if}\ x\in[2/3,1]
\end{array};
\right.
\eeqq
\beqq
f_2(x,y,z)=\left\{
\begin{array}{ll}
k_1y+k_2(x+z)& \hbox{if}\ y\in[0,1/3]\\
-k_1(y-1/3)+k_1/3+k_2(x+z)& \hbox{if}\ y\in[1/3,2/3]\\
k_1(y-2/3)+k_2(x+z)& \hbox{if}\ y\in[2/3,1]
\end{array};
\right.
\eeqq
\beqq
f_3(x,y,z)=k_3x,
\eeqq
and satisfies the following assumptions:
$k_1>2+2k_2$, $2k_2+k_1/3\leq1$, $k_1\geq10k_2>0$, and $0<k_3<(k_1-2k_2)/8$. It is easy to verify that this map satisfies all the assumptions in Theorem \ref{mainresult} with $N=1$.
\end{example}

Fix $k_1=2.4$, $k_2=0.08,$ and $k_3=0.25$. Figures 1 and 2 are the simulation diagrams with different initial values. In Figure 1, the initial value is taken as $(0.2,0.1,0.5)$. In Figure 2, the initial value is taken as $(0.5,0.5,0.5)$.
In Figure 1, the chaotic dynamical behavior is observed. In Figure 2, a ``regular" orbit is observed. From these simulations, we guess that there might exist two ergodic components. It is an interesting problem to prove or disprove the existence of at least two ergodic components.

\bigskip

\section*{Acknowledgments}
I would like to thank Professor Sheldon Newhouse for his encouragement, comments, and providing many useful references.

I devote this work to my grandmother, she moved to heaven in 2014. She will always be in my heart.

\newpage

\centerline{\Large \bf List of Figure Captions}

\baselineskip=16pt
\vspace{0.4 in} \noindent
\begin{enumerate}
\item[Figure 1.] The chaotic attractor of map in Example 4.1 with $k_1=2.4$, $k_2=0.08,$ and $k_3=0.25$, where the initial value is taken as $(0.2,0.1,0.5)$.

\item[Figure 2.] The chaotic attractor of map in Example 4.1 with $k_1=2.4$, $k_2=0.08,$ and $k_3=0.25$, where the initial value is taken as $(0.5,0.5,0.5)$.

\end{enumerate}

\begin{figure}[H]
\begin{center}
\scalebox{0.3 }{ \includegraphics{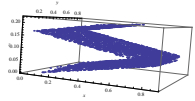}}
\renewcommand{\figure}{Fig.}
\caption{The chaotic attractor of map in Example 4.1 with $k_1=2.4$, $k_2=0.08,$ and $k_3=0.25$, where the initial value is taken as $(0.2,0.1,0.5)$.}\label{Figure1}
\end{center}
\end{figure}

\begin{figure}[H]
\begin{center}
\scalebox{0.3 }{ \includegraphics{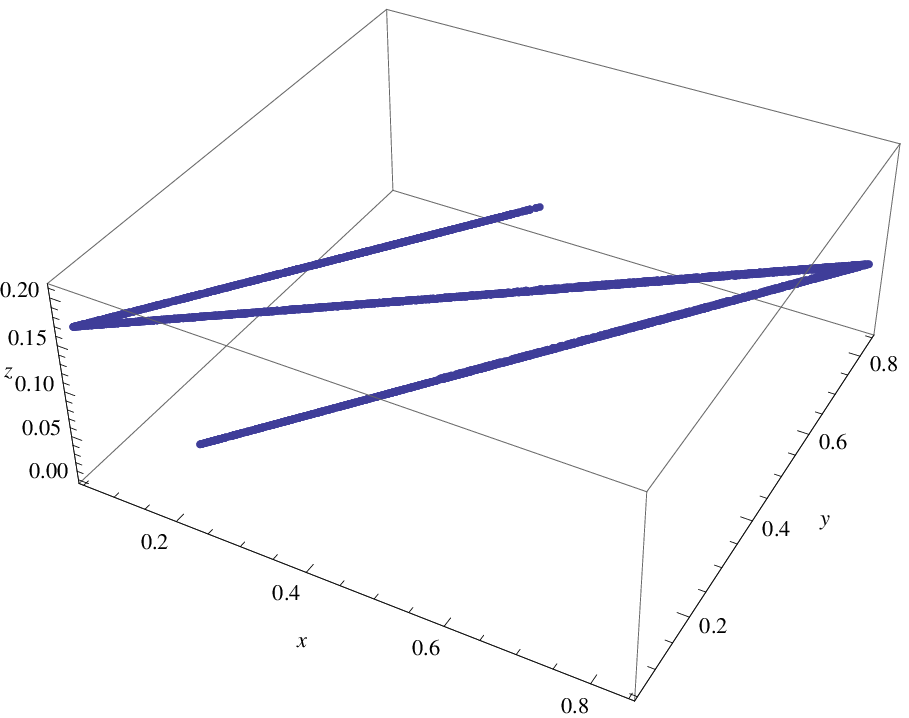}}
\renewcommand{\figure}{Fig.}
\caption{The chaotic attractor of map in Example 4.1 with $k_1=2.4$, $k_2=0.08,$ and $k_3=0.25$, where the initial value is taken as $(0.5,0.5,0.5)$.}\label{Figure2}
\end{center}
\end{figure}

\end{document}